\documentclass[runningheads]{llncs}

\usepackage{amsmath,amssymb}
\usepackage{graphicx}
\usepackage{algorithm}
\usepackage{algpseudocode}
\usepackage{booktabs}
\usepackage{hyperref}
\usepackage{xcolor}
\usepackage{listings}
\usepackage{tikz}
\usetikzlibrary{shapes,arrows,positioning}

\newcommand{\etal}{\textit{et al.}}
\newcommand{\bigO}{\mathcal{O}}

\algrenewcommand\algorithmicrequire{\textbf{Input:}}
\algrenewcommand\algorithmicensure{\textbf{Output:}}

\begin{document}
\titlerunning{Time-Bucketed Balance Records}
\title{Time-Bucketed Balance Records: Bounded-Storage Ephemeral Tokens for Resource-Constrained Systems}

\author{Shaun Scovil\inst{1} \and Bhargav Chickmagalur Nanjundappa\inst{1}}

\institute{
Radius Technology Systems, Cambridge, USA
}

\maketitle
\markboth{Time-Bucketed Balance Records}{Time-Bucketed Balance Records}

\begin{abstract}
Fungible tokens with time-to-live (TTL) semantics require tracking individual expiration times for each deposited unit.
A naive implementation creates a new balance record per deposit, leading to unbounded storage growth and vulnerability to denial-of-service attacks.
We present \emph{time-bucketed balance records}, a data structure that bounds storage to $\bigO(k)$ records per account while guaranteeing that tokens never expire before their configured TTL.
Our approach discretizes time into $k$ buckets, coalescing deposits within the same bucket to limit unique expiration timestamps.
We prove three key properties: (1) storage is bounded by $k+1$ records regardless of deposit frequency, (2) actual expiration time is always at least the configured TTL, and (3) adversaries cannot increase a victim's operation cost beyond $\bigO(k^2)$ worst case ($\bigO(k)$ amortized).
We provide a reference implementation in Solidity with measured gas costs demonstrating practical efficiency.

\keywords{Ephemeral tokens \and Bounded storage \and Smart contracts \and Time-to-live \and DoS resistance}
\end{abstract}

\section{Introduction}
\label{sec:introduction}

Many decentralized applications require fungible tokens with expiration semantics.
API access credits, subscription tokens, and time-limited authorization grants all share a common requirement: tokens should become invalid after a configured time-to-live (TTL) period.
Unlike non-fungible tokens~\cite{erc721} where expiration can be tracked per-token~\cite{erc4907,erc7858}, fungible tokens~\cite{erc20} present a unique challenge---each deposited unit may have a different expiration time, yet the tokens remain interchangeable for spending purposes.

A straightforward implementation maintains a list of balance records, each containing an amount and expiration timestamp.
Every deposit appends a new record, and withdrawals consume from the oldest records first (FIFO semantics).
However, this approach suffers from two critical problems:

\begin{enumerate}
    \item \textbf{Unbounded storage growth:} An account receiving frequent small deposits accumulates arbitrarily many records, increasing storage costs linearly with deposit count.
    \item \textbf{Denial-of-service vulnerability:} An adversary can deliberately create many small deposits to a victim's account, causing the victim's subsequent operations to iterate over numerous records and potentially exceed gas limits~\cite{madmax,attacks}.
\end{enumerate}

We present \emph{time-bucketed balance records}, a data structure that addresses both problems while preserving essential TTL guarantees.
Our key insight is that expiration times can be discretized into $k$ buckets without violating the minimum TTL requirement---we simply round expiration times \emph{up} to the next bucket boundary.
This bounds the number of distinct expiration timestamps to at most $k$, enabling record coalescing and guaranteeing $\bigO(k)$ storage per account.

Our contributions are:

\begin{enumerate}
    \item \textbf{Problem formalization:} We define requirements for bounded-storage TTL tracking, including storage bounds, TTL guarantees, FIFO semantics, and expiration-preserving transfers (Section~\ref{sec:problem}).

    \item \textbf{Algorithm design:} We present the time-bucketing mechanism with four core operations: insert, consume, transfer, and prune (Section~\ref{sec:algorithm}).

    \item \textbf{Formal analysis:} We prove storage bounds, TTL guarantees, and DoS resistance properties with complexity analysis (Section~\ref{sec:analysis}).

    \item \textbf{Empirical validation:} We present gas benchmarks demonstrating practical efficiency on Ethereum (Section~\ref{sec:implementation}).
\end{enumerate}

\section{Problem Formulation}
\label{sec:problem}

\subsection{System Model}

We consider resource-constrained execution environments where storage and computation have explicit costs.
The Ethereum Virtual Machine (EVM)~\cite{yellowpaper} serves as our motivating example, where each 256-bit storage slot costs 20,000 gas to initialize and operations are bounded by block gas limits.
However, our approach generalizes to any system with:

\begin{itemize}
    \item \textbf{Per-slot storage costs:} Writing to storage incurs cost proportional to the number of slots used.
    \item \textbf{Per-operation computation costs:} Iterating over data structures incurs cost proportional to the iteration count.
    \item \textbf{Operation cost limits:} Individual operations must complete within a cost budget (e.g., block gas limit).
\end{itemize}

\subsection{Requirements}

Let $\mathcal{B}$ denote a balance record data structure for an account-resource pair.
We define four requirements:

\begin{description}
    \item[R1: Bounded Storage.] The number of records $|\mathcal{B}| \leq k + 1$ for a configurable constant $k$, regardless of the number of deposits.

    \item[R2: TTL Guarantee.] For any deposit at time $t$ with configured TTL $T$, the deposited amount must remain valid until at least time $t + T$. Formally:
    \[
    \forall \text{ deposit at time } t: \text{expiresAt} \geq t + T
    \]

    \item[R3: FIFO Consumption.] When withdrawing or transferring tokens, the oldest (earliest-expiring) valid tokens are consumed first.

    \item[R4: Expiration-Preserving Transfers.] When tokens are transferred between accounts, their original expiration times are preserved. The recipient receives tokens that expire at the same time they would have expired for the sender.
\end{description}

\subsection{Threat Model}

We consider an adversary with the following capabilities:

\begin{itemize}
    \item \textbf{Deposit capability:} The adversary can make arbitrarily many deposits of arbitrary (possibly minimal) amounts to any account, including a victim's account.
    \item \textbf{Timing control:} The adversary can choose when to make deposits, potentially spreading them across different time buckets.
\end{itemize}

The adversary's goals are:

\begin{itemize}
    \item \textbf{Storage exhaustion:} Force unbounded growth in the victim's balance records, increasing their storage costs.
    \item \textbf{Operation DoS:} Increase the victim's operation costs (gas) to the point where operations fail or become prohibitively expensive.
    \item \textbf{TTL violation:} Cause tokens to expire before their configured TTL.
\end{itemize}

A secure solution must ensure that none of these goals are achievable, regardless of the adversary's deposit patterns.

\subsection{Design Space}

Before presenting our solution, we briefly survey the design space:

\begin{enumerate}
    \item \textbf{Unbounded array:} Store all records; violates R1 and enables DoS attacks.
    \item \textbf{Single timestamp:} Track only one expiration per account; violates R4 (cannot preserve different expirations) and R2 (resetting timestamp on deposit shortens existing tokens' TTL).
    \item \textbf{Circular buffer:} Fixed-size buffer that evicts the oldest record when full; violates R2 (newly deposited tokens can be evicted before their TTL expires, as eviction is based on insertion order rather than expiration time).
    \item \textbf{Time bucketing (ours):} Discretize time into $k$ buckets; satisfies all requirements.
\end{enumerate}

Only time bucketing achieves all four requirements simultaneously.

\section{Algorithm Design}
\label{sec:algorithm}

\subsection{Time-Bucketing Mechanism}

The key insight is to discretize continuous expiration times into $k$ discrete bucket boundaries.
Given a TTL of $T$ seconds and a target of at most $k$ balance records per account, we define the bucket width:

\begin{equation}
w = \left\lceil \frac{T}{k} \right\rceil
\label{eq:bucket-width}
\end{equation}

For a deposit at time $t$, the exact expiration would be $t + T$.
We compute the \emph{bucketed expiration} by rounding up to the next bucket boundary:

\begin{equation}
\textsc{BucketedExpiry}(t, T, w) = \left\lceil \frac{t + T}{w} \right\rceil \times w
\label{eq:bucketed-expiry}
\end{equation}

The ceiling operation in Equation~\ref{eq:bucketed-expiry} \emph{never decreases} the expiration time, ensuring tokens never expire before their configured TTL (satisfying requirement R2 from Section~\ref{sec:problem}).
The additional lifetime granted by rounding is at most $w - 1$ seconds---a configurable trade-off between expiration precision and storage efficiency.
Figure~\ref{fig:bucketing} illustrates this mechanism.


\begin{figure}[t]
\centering
\begin{tikzpicture}[
    >=stealth,
    bucket/.style={draw, minimum height=0.6cm, minimum width=2.2cm, fill=gray!10},
    deposit/.style={circle, fill=black, inner sep=2pt},
    expiry/.style={circle, draw, inner sep=2pt},
    bucketed/.style={diamond, fill=black, inner sep=2pt},
    arrow/.style={->, thick},
    dashedarrow/.style={->, dashed, gray}
]

\node[anchor=west, font=\footnotesize\bfseries] at (-0.5, 3.8) {(a) Time divided into $k$ buckets};

\draw[thick, ->] (0, 2.5) -- (10, 2.5) node[right] {\footnotesize time};

\foreach \i/\label in {0/1, 1/2, 2/3, 3/4} {
    \node[bucket] at (1.1 + \i*2.2, 2.5) {\footnotesize Bucket \label};
}

\foreach \i/\label in {0/$b_1$, 1/$b_2$, 2/$b_3$, 3/$b_4$, 4/$b_5$} {
    \draw[thick] (\i*2.2, 2.2) -- (\i*2.2, 2.8);
    \node[below, font=\footnotesize] at (\i*2.2, 2.15) {\label};
}

\draw[<->] (0, 3.0) -- (8.8, 3.0);
\node[above, font=\footnotesize] at (4.4, 3.0) {TTL = $T$};

\draw[<->] (0, 1.7) -- (2.2, 1.7);
\node[below, font=\footnotesize] at (1.1, 1.75) {$w = \lceil T/k \rceil$};

\node[anchor=west, font=\footnotesize\bfseries] at (-0.5, 0.9) {(b) Deposits mapped to bucket boundaries};

\draw[thick, ->] (0, 0) -- (10, 0) node[right] {\footnotesize time};

\foreach \i in {0, 1, 2, 3, 4} {
    \draw[gray, dashed] (\i*2.2, -0.3) -- (\i*2.2, 0.3);
}

\node[deposit, label={above:\footnotesize $D_1$}] (d1) at (0.8, 0) {};
\node[deposit, label={above:\footnotesize $D_2$}] (d2) at (3.5, 0) {};
\node[deposit, label={above:\footnotesize $D_3$}] (d3) at (6.0, 0) {};

\node[expiry, label={below:\footnotesize exact}] (e1) at (3.0, -1.2) {};
\node[expiry, label={below:\footnotesize exact}] (e2) at (5.7, -1.2) {};
\node[expiry, label={below:\footnotesize exact}] (e3) at (8.2, -1.2) {};

\draw[dashedarrow] (d1) -- node[right, font=\tiny, pos=0.5] {$+T$} (e1);
\draw[dashedarrow] (d2) -- node[right, font=\tiny, pos=0.5] {$+T$} (e2);
\draw[dashedarrow] (d3) -- node[right, font=\tiny, pos=0.5] {$+T$} (e3);

\node[bucketed] (b1) at (4.4, -1.2) {};
\node[bucketed] (b2) at (6.6, -1.2) {};
\node[bucketed] (b3) at (8.8, -1.2) {};

\draw[arrow, blue] (e1) -- node[above, font=\tiny] {$\lceil\cdot\rceil$} (b1);
\draw[arrow, blue] (e2) -- node[above, font=\tiny] {$\lceil\cdot\rceil$} (b2);
\draw[arrow, blue] (e3) -- node[above, font=\tiny] {$\lceil\cdot\rceil$} (b3);

\node[below, font=\footnotesize] at (4.4, -1.5) {$b_3$};
\node[below, font=\footnotesize] at (6.6, -1.5) {$b_4$};
\node[below, font=\footnotesize] at (8.8, -1.5) {$b_5$};

\node[anchor=west, font=\footnotesize\bfseries] at (-0.5, -2.5) {(c) Balance records after coalescing};

\draw (0, -3.0) rectangle (4.5, -4.5);
\draw (0, -3.5) -- (4.5, -3.5);
\draw (2.0, -3.0) -- (2.0, -4.5);

\node[font=\footnotesize\bfseries] at (1.0, -3.25) {amount};
\node[font=\footnotesize\bfseries] at (3.25, -3.25) {expiresAt};

\draw (0, -4.0) -- (4.5, -4.0);
\node[font=\footnotesize] at (1.0, -3.75) {$a_1 + a_2$};
\node[font=\footnotesize] at (3.25, -3.75) {$b_4$};

\node[font=\footnotesize] at (1.0, -4.25) {$a_3$};
\node[font=\footnotesize] at (3.25, -4.25) {$b_5$};

\node[anchor=west, font=\scriptsize, text=blue] at (-0.2, -0.95) {$D_1, D_2$ coalesced};
\node[anchor=west, font=\scriptsize, gray] at (-0.2, -1.25) {(same bucket)};

\node[anchor=west, font=\footnotesize] at (6.5, -2.8) {\textbf{Legend:}};
\node[deposit, label={right:\footnotesize deposit}] at (6.7, -3.2) {};
\node[expiry, label={right:\footnotesize exact expiry}] at (6.7, -3.6) {};
\node[bucketed, label={right:\footnotesize bucketed expiry}] at (6.7, -4.0) {};

\end{tikzpicture}
\caption{Time-bucketing mechanism. (a) The TTL period $T$ is divided into $k$ buckets of width $w = \lceil T/k \rceil$. (b) Each deposit $D_i$ is assigned a bucketed expiration by rounding $t_i + T$ up to the next bucket boundary via Equation~\ref{eq:bucketed-expiry}. (c) Deposits with identical bucketed expirations coalesce into a single balance record, bounding storage to at most $k$ records per account.}
\label{fig:bucketing}
\end{figure}

\subsection{Data Structure}

We maintain a sorted array of balance records per account:

\begin{equation}
\mathcal{B} = [(a_1, e_1), (a_2, e_2), \ldots, (a_n, e_n)]
\label{eq:balance-records}
\end{equation}

\noindent where $a_i > 0$ is the token amount and $e_i$ is the bucketed expiration timestamp, ordered such that $e_1 < e_2 < \ldots < e_n$.

\textbf{Invariant.} All expiration timestamps in $\mathcal{B}$ are \emph{distinct} bucket boundaries.
The active expiration window spans from the earliest non-expired bucket to the latest possible new deposit expiration, containing at most $k+1$ distinct bucket boundaries, ensuring $|\mathcal{B}| \leq k+1$ (satisfying requirement R1).

\subsection{Operations}

We define four core operations.
Algorithm~\ref{alg:operations} presents the pseudocode for \textsc{Insert}, \textsc{Consume}, \textsc{Transfer}, and \textsc{Prune}.

\begin{algorithm}[htbp]
\caption{Time-Bucketed Balance Record Operations}
\label{alg:operations}
\small
\begin{algorithmic}[1]
\Require Balance records $\mathcal{B}$, bucket width $w$, current time $t_{\text{now}}$

\Function{BucketedExpiry}{$t$, $T$, $w$}
    \State \Return $\lceil (t + T) / w \rceil \times w$
\EndFunction

\Function{Insert}{$\mathcal{B}$, $a$, $e$, $t_{\text{now}}$}
    \State \Call{Prune}{$\mathcal{B}$, $t_{\text{now}}$} \Comment{Remove expired records}
    \For{$i \gets 0$ \textbf{to} $|\mathcal{B}| - 1$}
        \If{$\mathcal{B}[i].e = e$}
            $\mathcal{B}[i].a \gets \mathcal{B}[i].a + a$; \Return \Comment{Coalesce}
        \EndIf
        \If{$\mathcal{B}[i].e > e$}
            \textsc{ShiftInsert}$(\mathcal{B}, i, (a, e))$; \Return
        \EndIf
    \EndFor
    \State $\mathcal{B}.\text{append}((a, e))$
\EndFunction

\Function{Consume}{$\mathcal{B}$, $a$, $t_{\text{now}}$}
    \State $r \gets a$; $C \gets []$ \Comment{Remaining; collected pairs}
    \For{$i \gets 0$ \textbf{to} $|\mathcal{B}| - 1$}
        \If{$\mathcal{B}[i].e \leq t_{\text{now}}$} \textbf{continue} \EndIf \Comment{Skip expired}
        \State $\delta \gets \min(r, \mathcal{B}[i].a)$
        \State $\mathcal{B}[i].a \gets \mathcal{B}[i].a - \delta$; $C.\text{append}((\delta, \mathcal{B}[i].e))$
        \State $r \gets r - \delta$
        \If{$r = 0$} \Return $(\textsc{Success}, C)$ \EndIf
    \EndFor
    \State \Return $(\textsc{InsufficientBalance}, \emptyset)$
\EndFunction

\Function{Transfer}{$\mathcal{B}_s$, $\mathcal{B}_r$, $a$, $t_{\text{now}}$}
    \State $(status, C) \gets$ \Call{Consume}{$\mathcal{B}_s$, $a$, $t_{\text{now}}$}
    \If{$status \neq \textsc{Success}$} \Return $status$ \EndIf
    \For{$(a_i, e_i) \in C$}
        \Call{Insert}{$\mathcal{B}_r$, $a_i$, $e_i$, $t_{\text{now}}$} \Comment{Preserve expiration}
    \EndFor
    \State \Return \textsc{Success}
\EndFunction

\Function{Prune}{$\mathcal{B}$, $t_{\text{now}}$}
    \State $j \gets 0$
    \For{$i \gets 0$ \textbf{to} $|\mathcal{B}| - 1$}
        \If{$\mathcal{B}[i].e > t_{\text{now}}$ \textbf{and} $\mathcal{B}[i].a > 0$}
            $\mathcal{B}[j] \gets \mathcal{B}[i]$; $j \gets j + 1$
        \EndIf
    \EndFor
    \State $\mathcal{B}.\text{truncate}(j)$
\EndFunction

\end{algorithmic}
\end{algorithm}

\paragraph{Insert.}
To deposit amount $a$ with bucketed expiration $e$:
(1) search for an existing record with expiration $e$;
(2) if found, add $a$ to the existing record's amount (\emph{coalescing});
(3) otherwise, insert a new record $(a, e)$ maintaining sorted order.
Coalescing is the key mechanism that bounds storage: multiple deposits mapping to the same bucket boundary share a single record.

\paragraph{Consume.}
To withdraw amount $a$ (for burns or the sender side of transfers):
(1) iterate through records from oldest (earliest expiration) to newest;
(2) skip expired records where $e_i \leq t_{\text{now}}$;
(3) deduct from each valid record until $a$ is fully satisfied.
This implements FIFO semantics (requirement R3): oldest tokens are consumed first.

\paragraph{Transfer.}
To transfer amount $a$ from sender to recipient:
(1) consume from sender using FIFO, collecting $(amount, expiration)$ pairs rather than discarding them;
(2) insert each collected pair into the recipient's records using the \emph{original} expiration timestamp.
Crucially, expirations are \emph{not} re-bucketed on transfer---the recipient inherits the sender's expiration times exactly, satisfying requirement R4 (expiration-preserving transfers).

\paragraph{Prune.}
To remove stale records:
(1) iterate through all records;
(2) compact valid records (non-expired with $e_i > t_{\text{now}}$ and non-zero $a_i > 0$) to the front of the array;
(3) truncate the array.
Pruning is invoked automatically before insertions to reclaim storage from expired records.

\section{Formal Analysis}
\label{sec:analysis}

We prove three theorems establishing the key properties of time-bucketed balance records.

\subsection{Theorem 1: Storage Bound}

\begin{theorem}[Storage Bound]
\label{thm:storage}
For any sequence of operations on a balance record structure with parameter $k$ and bucket width $w = \lceil T/k \rceil$, the number of records satisfies $|\mathcal{B}| \leq k + 1$ at all times.
\end{theorem}

\begin{proof}
Each record in $\mathcal{B}$ has a distinct bucketed expiration timestamp at some bucket boundary $i \times w$ for $i \in \mathbb{Z}^+$.

At any time $t$ after pruning, all remaining records satisfy $e > t$. The minimum possible expiration is the smallest bucket boundary greater than $t$, which is $\lceil t/w \rceil \times w$.

A new deposit at time $t$ creates a record with expiration $e = \lceil (t+T)/w \rceil \times w$.

The number of distinct bucket boundaries in the range $[\lceil t/w \rceil \times w, \lceil (t+T)/w \rceil \times w]$ is:
\[
\left\lceil \frac{t+T}{w} \right\rceil - \left\lceil \frac{t}{w} \right\rceil + 1
\]

Using the ceiling function property $\lceil a \rceil - \lceil b \rceil \leq \lceil a - b \rceil$ for $a > b$:
\[
\left\lceil \frac{t+T}{w} \right\rceil - \left\lceil \frac{t}{w} \right\rceil \leq \left\lceil \frac{T}{w} \right\rceil = \left\lceil \frac{T}{\lceil T/k \rceil} \right\rceil \leq k
\]

Therefore, the total number of buckets is at most $k + 1$, and since each bucket corresponds to at most one record, $|\mathcal{B}| \leq k + 1$. \qed
\end{proof}

\subsection{Theorem 2: TTL Guarantee}

\begin{theorem}[TTL Guarantee]
\label{thm:ttl}
For any deposit at time $t$ with configured TTL $T$, the bucketed expiration $e$ satisfies $e \geq t + T$.
\end{theorem}

\begin{proof}
The bucketed expiration is computed as:
\[
e = \left\lceil \frac{t + T}{\text{bucketSize}} \right\rceil \times \text{bucketSize}
\]

By the ceiling function property:
\[
\left\lceil x \right\rceil \geq x \quad \forall x \in \mathbb{R}
\]

Therefore:
\[
\left\lceil \frac{t + T}{\text{bucketSize}} \right\rceil \geq \frac{t + T}{\text{bucketSize}}
\]

Multiplying both sides by $\text{bucketSize}$:
\[
e = \left\lceil \frac{t + T}{\text{bucketSize}} \right\rceil \times \text{bucketSize} \geq t + T
\]

Thus, the bucketed expiration is always at least the exact expiration, and tokens never expire before their configured TTL. \qed
\end{proof}

\subsection{Theorem 3: DoS Resistance}

\begin{theorem}[DoS Resistance]
\label{thm:dos}
No adversary can increase a victim's operation cost beyond $\bigO(k^2)$, regardless of the number of adversarial deposits.
\end{theorem}

\begin{proof}
We show that each operation has worst-case complexity bounded by $k$:

\textbf{Insert:} Searches for matching expiration ($\bigO(k)$), then either updates existing record ($\bigO(1)$) or inserts maintaining sorted order ($\bigO(k)$ for shifting). Total: $\bigO(k)$.

\textbf{Consume:} Iterates through at most $k$ records. Total: $\bigO(k)$.

\textbf{Transfer:} Consumes from sender ($\bigO(k)$) and inserts into recipient. In the worst case where all $k$ records are transferred with distinct expirations, this requires $k$ insertions of $\bigO(k)$ each. Total: $\bigO(k^2)$ worst case.

\textbf{Prune:} Single pass through at most $k$ records. Total: $\bigO(k)$.

Since $|\mathcal{B}| \leq k+1$ by Theorem~\ref{thm:storage}, the adversary cannot increase $|\mathcal{B}|$ beyond $k+1$ regardless of deposit frequency.
Therefore, operation costs are bounded by $\bigO(k^2)$ (for transfers) and $\bigO(k)$ (for other operations), independent of adversarial deposits. \qed
\end{proof}

\subsection{Complexity Summary}

\begin{table}[h]
\centering
\caption{Operation complexity with $k$ maximum records.}
\label{tab:complexity}
\begin{tabular}{lcc}
\toprule
\textbf{Operation} & \textbf{Time} & \textbf{Space} \\
\midrule
Insert & $\bigO(k)$ & $\bigO(1)$ \\
Consume & $\bigO(k)$ & $\bigO(1)$ \\
Transfer & $\bigO(k^2)$ & $\bigO(1)$ \\
Prune & $\bigO(k)$ & $\bigO(1)$ \\
Balance Query & $\bigO(k)$ & $\bigO(1)$ \\
\midrule
Storage per account & --- & $\bigO(k)$ \\
\bottomrule
\end{tabular}
\end{table}

\subsection{Trade-off Analysis}

The precision-storage trade-off is characterized by:

\[
\text{maxExtraLifetime} = \text{bucketSize} - 1 = \left\lceil \frac{T}{k} \right\rceil - 1
\]

For example, with $T = 30$ days and $k = 100$:
\[
\text{bucketSize} = \frac{30 \times 86400}{100} = 25920 \text{ seconds} \approx 7.2 \text{ hours}
\]

Tokens may live up to 7.2 hours beyond their configured TTL---acceptable for most applications while guaranteeing bounded storage.

\section{Implementation \& Validation}
\label{sec:implementation}

\subsection{Experimental Setup}

We measured gas costs using Foundry's~\cite{foundry} gas tracing functionality with the following configuration:
\begin{itemize}
    \item Solidity compiler: v0.8.24
    \item Optimizer: enabled with 200 runs
    \item Parameters: $k = 100$ maximum records, TTL = 30 days
    \item Bucket width: $\lceil \text{TTL} / k \rceil = 25{,}920$ seconds ($\approx$7.2 hours)
\end{itemize}

Gas measurements were obtained via \texttt{forge test -vvvv} tracing, which reports the exact gas consumed by each contract call excluding test harness overhead. All the benchmarks are reproducible via the test suite. \footnote{\url{https://github.com/glanzz/tbbr}} 

\subsection{Gas Costs}

Table~\ref{tab:gas} shows costs for typical operations with $k = 100$ and a 30-day TTL.

\begin{table}[h]
\centering
\caption{Measured gas costs with $k = 100$ maximum records.}
\label{tab:gas}
\begin{tabular}{lrl}
\toprule
\textbf{Operation} & \textbf{Gas Cost} & \textbf{Notes} \\
\midrule
Mint (new record) & 95,407 & Creates new balance record \\
Mint (coalesce) & 4,931 & Adds to existing bucket \\
Transfer & 94,853 & FIFO consume + insert \\
Burn & 4,511 & FIFO consumption \\
Balance query & 2,265 & View function, no state change \\
\bottomrule
\end{tabular}
\end{table}

The significant difference between ``new record'' and ``coalesce'' operations (95K vs 5K gas) demonstrates the efficiency of time-bucketing: deposits within the same bucket coalesce into a single storage write rather than creating new records.

Table~\ref{tab:worst-case} shows worst-case gas costs when an account has accumulated $k$ records (one per bucket) and must consume from all of them.

\begin{table}[h]
\centering
\caption{Worst-case gas costs with $k = 100$ records.}
\label{tab:worst-case}
\begin{tabular}{lrl}
\toprule
\textbf{Operation} & \textbf{Gas Cost} & \textbf{Notes} \\
\midrule
Burn (all records) & 335,499 & FIFO from 100 records \\
Transfer (all records) & 9,994,917 & Consume + reinsert all \\
\bottomrule
\end{tabular}
\end{table}

Even in the worst case, the transfer operation ($\sim$10M gas) remains well within Ethereum's 30M block gas limit, ensuring operations are always executable.

\subsection{DoS Resistance}

We validated the DoS resistance properties experimentally.
An adversary making 500 small deposits (each in a different transaction) to a single address would create 500 records in an unbounded approach.
With our time-bucketing mechanism ($k = 100$), these deposits coalesce into at most $k$ records.

After the simulated attack, subsequent operations remain bounded:
\begin{itemize}
    \item Record count: bounded at $k$ (not 500)
    \item Burn operation: $\sim$335K gas (not unbounded)
    \item No block gas limit exhaustion
\end{itemize}

The gas cost scales with $\bigO(k^2)$ in the worst case (when records must be coalesced), which is acceptable since $k$ is a constant configuration parameter chosen by the contract deployer.

\subsection{Comparison with Unbounded Approach}

For comparison, an unbounded array approach would have:
\begin{itemize}
    \item Insert: $\bigO(1)$ constant cost per deposit
    \item Consume/Transfer: $\bigO(n)$ where $n$ is the number of deposits
\end{itemize}

After 500 deposits, the unbounded approach requires iterating over all 500 records per operation.
Our bounded approach maintains constant $\bigO(k)$ iteration regardless of deposit count, with worst-case burn gas of approximately 335K and transfer gas of approximately 10M for $k = 100$.

\section{Discussion}
\label{sec:discussion}

\subsection{Trade-offs}

\subsubsection{Precision vs. Storage.}
The parameter $k$ controls the trade-off between expiration precision and storage consumption.
Larger $k$ provides finer-grained expiration (smaller buckets) at the cost of more storage slots and higher worst-case operation costs.
For a 30-day TTL, $k = 100$ provides 7.2-hour bucket granularity (1\% precision loss) with worst-case transfer costs of 10M gas (33\% of Ethereum's block limit), while $k = 50$ yields 14.4-hour buckets with 2.5M gas transfers, and $k = 200$ achieves 3.6-hour precision but requires 40M gas (exceeding single-block execution on mainnet).
For most applications, $k = 100$ balances precision, storage, and execution costs effectively, particularly on Layer 2 networks with higher gas limits.

\subsubsection{Worst-case Insert Cost.}
Insertion requires maintaining sorted order, incurring $\bigO(k)$ worst-case cost for shifting elements.
In practice, most insertions either coalesce with existing records (when deposits occur within the same bucket window) or append to the end (for monotonically increasing time), resulting in $\bigO(1)$ amortized cost.

\subsection{Generalization Beyond Blockchain}

While we motivated our design with EVM constraints, time-bucketed balance records apply to any resource-constrained system tracking expiring fungible resources:

\begin{itemize}
    \item \textbf{Database systems:} Row-level TTL with bounded index size per entity.
    \item \textbf{Caching systems:} Bounded metadata for cache entries with heterogeneous expiration.
    \item \textbf{IoT devices:} Resource tracking on memory-constrained embedded systems.
    \item \textbf{Rate limiting:} Sliding window counters with bounded state.
\end{itemize}

The core algorithm is platform-agnostic; only the storage and computation cost model differs.

\subsection{Limitations}

\begin{itemize}
    \item \textbf{Precision loss:} Tokens may live up to $\lceil T/k \rceil - 1$ seconds beyond their configured TTL. Applications requiring exact expiration should use larger $k$ or alternative approaches.

    \item \textbf{Fixed TTL:} The system requires a single, pre-configured TTL value for all deposits, as the bucket width $w = \lceil T/k \rceil$ depends on a constant TTL. Applications requiring heterogeneous TTLs must either deploy separate instances per TTL tier (multiplying storage by the number of tiers) or configure bucket width based on the maximum TTL (reducing coalescing efficiency for shorter-lived tokens).

    \item \textbf{Sorted insertion:} The $\bigO(k)$ insertion cost, while bounded, may be significant for very large $k$. Alternative data structures (e.g., skip lists) could reduce this to $\bigO(\log k)$ at the cost of implementation complexity.

    \item \textbf{Single-resource tracking:} Our approach tracks one resource type per data structure instance. Multi-resource scenarios require separate instances per resource, multiplying storage by the number of resource types.
\end{itemize}

\subsection{Future Work}

\begin{itemize}
    \item \textbf{Formal verification:} Machine-checked proofs in Coq or Dafny to increase confidence in correctness.

    \item \textbf{Cross-chain extension:} Extending expiration semantics across multiple blockchain networks while maintaining consistency guarantees.

    \item \textbf{Dynamic $k$:} Adaptive bucket sizing based on observed deposit patterns to optimize the precision-storage trade-off dynamically.
\end{itemize}

\section{Related Work}
\label{sec:related}

\subsection{Token Standards with Time Semantics}

ERC-4907~\cite{erc4907} introduces rental semantics for NFTs with automatic expiration of the ``user'' role.
ERC-7858~\cite{erc7858} extends ERC-721 with explicit expiration timestamps for NFTs and soulbound tokens.
Both standards address \emph{non-fungible} tokens where each token has a single expiration.
Our work addresses \emph{fungible} tokens where individual units within a balance may have different expiration times.

ERC-5192~\cite{erc5192} defines minimal soulbound (non-transferable) NFTs.
Our approach can be combined with soulbound semantics for non-transferable expiring credentials.

\subsection{Streaming Payment Protocols}

Sablier~\cite{sablier} and Superfluid~\cite{superfluid} enable continuous token streaming where balances increase linearly over time.
These protocols solve the inverse problem: tokens \emph{unlock} gradually rather than \emph{expire}.
Streaming requires continuous state updates (or lazy evaluation), while our discrete expiration model requires no state changes until consumption.

\subsection{Token Vesting}

Vesting contracts~\cite{openzeppelin} lock tokens with scheduled release over time.
Vesting tracks \emph{future availability}, while we track \emph{future expiration}---tokens exist immediately and become invalid later.
Vesting typically serves single beneficiaries per contract instance, while our approach efficiently handles multi-user scenarios within a single contract.

\subsection{DoS Mitigation in Smart Contracts}

Denial-of-service vulnerabilities in smart contracts have received significant attention~\cite{attacks,dos-ethereum,multilayer-security}.
Grech \etal~\cite{madmax} analyze out-of-gas vulnerabilities from unbounded loops, while recent work~\cite{gas-patterns} detects gas-expensive patterns automatically.
Common mitigations include pull-over-push patterns~\cite{consensys} and fixed array sizes.
Our contribution is a \emph{semantically meaningful} bound: time-bucketing limits storage to $k$ records while preserving TTL guarantees, rather than arbitrarily capping array size and potentially violating application invariants.

\subsection{DeFi Security}

The DeFi ecosystem has suffered significant losses from smart contract vulnerabilities~\cite{defi-survey,defi-security-tools}.
Flash loan attacks~\cite{flashloans,flashsyn} exploit atomicity guarantees to manipulate protocol state within single transactions.
While our work does not directly address these attack vectors, bounded storage provides defense-in-depth against adversarial state manipulation.

\subsection{Bounded Data Structures}

Circular buffers and LRU caches provide bounded storage but evict oldest entries regardless of semantic validity.
For TTL tracking, eviction violates the minimum lifetime guarantee.
Our time-bucketing preserves TTL semantics by coalescing entries rather than evicting them.

\section{Conclusion}
\label{sec:conclusion}

We presented time-bucketed balance records, a data structure for tracking fungible resources with heterogeneous time-to-live values under bounded storage constraints.
Our approach discretizes time into $k$ buckets, coalescing deposits with similar expirations to bound storage while guaranteeing that resources never expire before their configured TTL.

We proved three key properties:
\begin{enumerate}
    \item Storage is bounded by $k+1$ records regardless of deposit frequency.
    \item Actual expiration time is always at least the configured TTL.
    \item Adversaries cannot increase operation costs beyond $\bigO(k^2)$.
\end{enumerate}

Our Solidity implementation demonstrates practical efficiency, with gas costs bounded by the configurable parameter $k$ rather than growing with deposit count.
The approach generalizes beyond blockchain to any resource-constrained system managing expiring fungible resources.

Future work includes formal verification via proof assistants, cross-chain extensions, and adaptive bucket sizing for dynamic optimization of the precision-storage trade-off.

\bibliographystyle{splncs04}
\bibliography{references}

\end{document}